\documentclass[copyright,creativecommons]{eptcs}
\providecommand{\event}{VPT 2015} 
\usepackage{breakurl}             

\usepackage{amsmath}
\usepackage{amsthm}
\usepackage{amsfonts}

\usepackage{tikz}
\usepackage{tikz-qtree}

\newcommand{\lambdat}{\mathit{lambda}}
\newcommand{\unlambdat}{\mathit{unlambda}}

\newcommand{\Label}{\mathit{Label}}
\newcommand{\Var}{\mathit{Var}}
\newcommand{\Abs}{\mathit{Abs}}

\newcommand{\Fun}{\mathit{FUN}}

\newcommand{\defeq}{\stackrel{\mathit{def}}{=}}
\newcommand{\ang}[1]{\langle {#1} \rangle}

\newcommand{\Bool}{\mathit{Bool}}

\newtheorem{lemma}{Lemma}
\newtheorem{theorem}{Theorem}
\newtheorem{corollary}{Corollary}

\title{Control Flow Analysis for SF Combinator Calculus}
\author{Martin Lester
\institute{Department of Computer Science, University of Oxford, Oxford, UK}
\email{martin.lester@cs.ox.ac.uk}
}
\def\titlerunning{CFA for SF-Calculus}
\def\authorrunning{M. M. Lester}
\begin{document}
\maketitle

\begin{abstract}
Programs that transform other programs often require access to the
internal structure of the program to be transformed. This is at odds with
the usual extensional view of functional programming, as embodied by the
lambda calculus and SK combinator calculus. The recently-developed SF
combinator calculus offers an alternative, intensional model of computation
that may serve as a foundation for developing principled languages in which
to express intensional computation, including program transformation. Until
now there have been no static analyses for reasoning about or verifying
programs written in SF-calculus. We take the first step towards remedying
this by developing a formulation of the popular control flow analysis 0CFA
for SK-calculus and extending it to support SF-calculus. We prove its
correctness and demonstrate that the analysis is invariant under the usual
translation from SK-calculus into SF-calculus.

\end{abstract}

\section{Introduction}
In order to reason formally about the behaviour of program transformations,
we must simultaneously consider the semantics of both the program being transformed
and the program performing the transformation.
In most languages, program code is not a first-class citizen:
typically, the code and its manipulation and execution
are encoded in an ad-hoc manner using the standard datatypes of the language.
Consequently, whenever we want to reason formally about program transformation,
we must first formalise the link between the encoding of the code and its semantics.

This is unsatisfying and it would be desirable to develop better techniques for
reasoning about program transformation in a more general way.
In order to do this, we must develop techniques for reasoning about
programs that manipulate other programs.
That is, we need to develop techniques for
reasoning about and verifying uses of \emph{metaprogramming}.
Metaprogramming can be split into \emph{extensional} and \emph{intensional} uses:
extensional metaprogramming involves joining together pieces of program code,
treating the code as a ``black box'';
intensional metaprogramming allows inspection and manipulation of
the internal structure of code values.

Unfortunately, as support for metaprogramming is relatively poor in most programming languages,
its study and verification is often not a priority.
In particular, the $\lambda$-calculus,
which is often thought of as the theoretical foundation of functional programming languages,
does not allow one to express programs that can distinguish between two
extensionally equal expressions with different implementations,
or indeed to manipulate the internal structure of expressions in any way.

However, the SF combinatory calculus~\cite{DBLP:journals/jsyml/Given-WilsonJ11}
does allow one to express such programs.
SF-calculus is a formalism similar to the familiar SK combinatory calculus,
which is itself similar to $\lambda$-calculus,
but avoids the use of variables and hence the complications of substitution and renaming.
SF-calculus replaces the $K$ of SK-calculus with a factorisation combinator $F$
that allows one to deconstruct or \emph{factorise} program terms in certain normal forms.
Thus it may be a suitable theoretical foundation for
programming languages that support intensional metaprogramming.

There has been some recent work on verification of programs using extensional metaprogramming,
mainly in languages that only allow the composition and execution
of well-formed code templates~\cite{DBLP:conf/lpar/BergerT10,DBLP:conf/popl/ChoiAYT11}.
In contrast, verification of intensional metaprogramming has been comparatively neglected.

There do not yet appear to be any \emph{static analyses} for verifying properties of SF-calculus programs.
We rectify this by formulating the popular analysis \emph{0CFA}~\cite{DBLP:conf/pldi/Shivers88} for SF-calculus.
We prove its correctness and argue, with reference to a new formulation of 0CFA for SK-calculus,
why it is appropriate to call the analysis 0CFA.
This provides the groundwork for more expressive analyses of programs that manipulate programs.

We begin in Section 2 by reviewing SK-calculus and 0CFA for $\lambda$-calculus;
we also present a summary of SF-calculus.
In Section 3, we reformulate 0CFA for SK-calculus and prove its correctness;
this guides our formulation and proof of 0CFA for SF-calculus in Section 4.
We discuss the precision of our analysis in Section 5
and compare it with some related work in Section 6.
We conclude by suggesting some future research directions in Section 7.

\section{Preliminaries}
\subsection{0CFA for Lambda Calculus}

0CFA~\cite{DBLP:conf/pldi/Shivers88} is a popular form of \emph{Control Flow Analysis}.
It is flow insensitive and context insensitive,
but precise enough to be useful for many applications,
including guiding compiler optimisations~\cite{DBLP:journals/csur/Midtgaard12,DBLP:conf/icfp/BergstromFLRS14},
providing autocompletion hints in IDEs
and noninterference analysis~\cite{DBLP:conf/csfw/0001OS13}.
It is perhaps the simplest static analysis that handles higher order functions,
which are a staple of functional programming.

Let us consider 0CFA for the $\lambda$-calculus.
0CFA can be formulated in many ways.
Following Nielson and others,
we present it as a system of constraints~\cite{DBLP:books/daglib/0098888}.
Suppose we wish to analyse a program $e$.
We begin by assigning a unique label $l$ (drawn from a set $\Label$)
to every subexpression
(variable, application or $\lambda$-abstraction) in $e$.
(Reusing labels does not invalidate the analysis,
and indeed this is done deliberately in proving its correctness,
but it does reduce its precision.)
We write $e^l$ to make explicit reference to the label $l$ on $e$.
We often write applications infix as $e_1 @^l e_2$ rather than $(e_1\ e_2)^l$
to make reference to their labels clearer.
We follow the usual convention that application associates to the left,
so $f\ g\ x$ (or $f @ g @ x$) means $(f\ g)\ x$ and not $f\ (g\ x)$.

Next, we generate constraints on a function $\Gamma$ by
recursing over the structure of $e$,
applying the rules shown in Figure~\ref{fig:0cfa-lambda}.
Finally, we solve the constraints to produce
$\Gamma : \Label \uplus \Var \rightarrow \mathcal{P}(\Abs)$ that indicates,
for each position indicated by a subexpression label $l$ or variable $x$,
an over-approximation of all possible expressions that may occur in that position during evaluation.
Abstractly represented values $v$ have the form $\Fun(x, l)$,
indicating any expression $\lambda x . e^l$ that binds the variable $x$ to a body with label $l$.
We say that $\Gamma \models e$ if $\Gamma$ is a solution for the constraints generated over $e$.

\begin{figure}
\[
\begin{array}{lrcccl}
\mbox{Labels} & \Label & \ni & l & & \\
\mbox{Variables} & \Var & \ni & x & & \\
\mbox{Labelled Expressions} & & & e & ::= & x^l \mid e_1 @^l e_2 \mid \lambda^l x . e \\
\mbox{Abstract Values} & \Abs & \ni & v & ::= & \Fun(x, l) \\
\mbox{Abstract Environment} & & & \Gamma & : & \Label \uplus \Var \rightarrow \mathcal{P}(\Abs)
\end{array}
\]

\[
\begin{array}{lcl}
\Gamma \models x^l & \iff & \Gamma(x) \subseteq \Gamma(l) \\
\Gamma \models \lambda^{l_1} x . e^{l_2} & \iff &
    \Gamma \models e^{l_2} \wedge \Fun(x, l_2) \in \Gamma(l_1) \\
\Gamma \models e^{l_1}_1 @^{l} e^{l_2}_2 & \iff &
    \Gamma \models e^{l_1}_1 \wedge \Gamma \models e^{l_2}_2 \wedge
    (\forall \Fun(x, l_3) \in \Gamma(l_1) . \Gamma(l_2) \subseteq \Gamma(x) \wedge \Gamma(l_3) \subseteq \Gamma(l))
\end{array}
\]

\caption{0CFA for $\lambda$-calculus}
\label{fig:0cfa-lambda}
\end{figure}

The intuition behind the rules for $\Gamma \models e$ is as follows:
\begin{itemize}
\item \emph{If $e = x^l$:} $\Gamma(x)$ must over-approximate the values that can be bound to $x$.
\item \emph{If $e = \lambda^{l_1} x. {e'}^{l_2}$:} A $\lambda$-expression is represented abstractly
as $\Fun(x, l_2)$
by the variable it binds and the label on its body.
Furthermore, its subexpressions must be analysed.
\item \emph{If $e = e_1^{l_1} @^l e_2^{l_2}$:} For any application, consider all the functions that may occur on the left
and all the arguments that may occur on the right.
Each argument may be bound to any of the variables in the functions.
The result of the application may be the result of any of the function bodies.
Again, all subexpressions of the application must be analysed.
\end{itemize}

In order to argue about the soundness of the analysis,
we must first formalise what $\Gamma$ means.
We can do this via a \emph{labelled semantics} for $\lambda$-calculus
that extends the usual rules for evaluating $\lambda$-calculus expressions
to labelled expressions.
We can then prove a coherence theorem~\cite{DBLP:conf/esop/WandW02}:
if $\Gamma \models e^l$ and (in the labelled semantics)
$e^l \rightarrow {e'}^{l'}$, then $\Gamma \models {e'}^{l'}$
and $\Gamma(l') \subseteq \Gamma(l)$.
In fact, by induction on the length of derivations of $\rightarrow^*$,
this is also true for $e^l \rightarrow^* {e'}^{l'}$.
Note in particular that, as $\Gamma(l') \subseteq \Gamma(l)$,
$\Gamma(l)$ gives a sound over-approximation to
the subexpressions that may occur at the top level
at any point during evaluation.

As a concrete example, consider the $\lambda$-expression
$(\lambda^1 x . x^0) @^4 (\lambda^3 y . y^2)$,
which applies the identity function to itself.
We have chosen $\Label$ to be the natural numbers $\mathbb{N}$.
A solution for $\Gamma$ is:
\[
\begin{array}{ccccc}
\Gamma(x) = \Gamma(0) = \Gamma(3) = \Gamma(4) = \{ \Fun(y, 2) \} & &
\Gamma(1) = \{ \Fun(x, 0) \} & &
\Gamma(y) = \{ \}
\end{array}
\]
In particular, this correctly tells us that the result of evaluating the expression
is abstracted by $\Fun(y, 2)$; that is, the identity function with body labelled $2$ that binds $y$.


Note that the constraints on $\Gamma$ may easily be solved by:
initialising every $\Gamma(l)$ to be empty;
iteratively considering each unsatisfied constraint in turn
and enlarging some $\Gamma(l)$ in order to satisfy it;
stopping when a fixed point is reached and all constraints are satisfied.
Done naively, this takes time $\mathcal{O}(n^5)$
for a program of size $n$~\cite{DBLP:books/daglib/0098888}.
With careful ordering of the consideration of constraints,
this improves to $\mathcal{O}(n^3)$.
The best known algorithm for 0CFA uses
an efficient representation of the sets in $\Gamma$
to achieve $\mathcal{O}(n^3/\log n)$ complexity~\cite{DBLP:conf/popl/Chaudhuri08}.
Van Horn and Mairson showed that,
for linear programs (in which each bound variable occurs exactly once),
0CFA gives the same result as actually evaluating the program;
hence it is PTIME-complete~\cite{DBLP:conf/sas/HornM08}.

0CFA has been the inspiration for many other analyses.
For example, $k$-CFA adds $k$ levels of context to distinguish between
uses of the same function from different points within a program.
This improves precision, but at the cost of making the analysis EXPTIME-complete,
even for $k = 1$~\cite{DBLP:conf/icfp/HornM08}.
CFA2 similarly tries to use context to improve precision,
but via a pushdown abstraction, which remains practical~\cite{DBLP:journals/corr/abs-1102-3676}.

\subsection{SK Combinatory Calculus}

\begin{figure}[p]
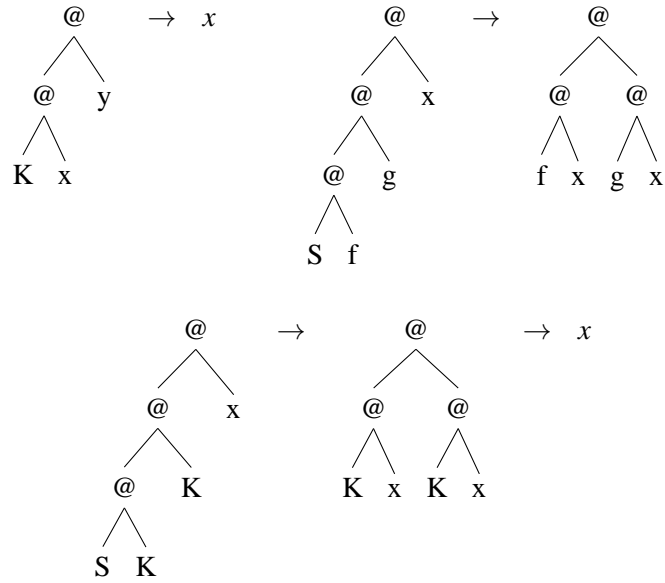

\[
\begin{array}{ccccccc}
\Tree [.@ [.@ K x ] y ] & \rightarrow & x &
\hspace{1em} &
\Tree [.@ [.@ [.@ S f ] g ] x ] & \rightarrow & \Tree [.@ [.@ f x ] [.@ g x ] ]
\end{array}
\]
\vspace{1ex}
\[
\begin{array}{ccccc}
\Tree [.@ [.@ [.@ S K ] K ] x ] & \rightarrow & \Tree [.@ [.@ K x ] [.@ K x ] ]
& \rightarrow & x
\end{array}
\]
\caption{Terms of SK-calculus viewed as trees.
Above: the reduction rules for $S$ and $K$.
Below: evaluation of the identity function $S\ K\ K$.}
\label{fig:sk-tree}
\end{figure}

\emph{Combinatory logic} is a Turing-powerful formalism for computation
that is similar in style to the $\lambda$-calculus,
but without bound variables and the associated complications
of capture-avoiding substitution and $\alpha$-conversion~\cite{Hindley:2008:LCI:1388400}.
From the perspective of term rewriting systems,
a combinator is a named constant $C$ with an associated rewrite rule
$C\ \overline{x} \rightarrow t(\overline{x})$,
where $\overline{x}$ is a sequence of variables of fixed length
and $t(\overline{x})$ is term built from the variables in $\overline{x}$ using application;
that is, $t(\overline{x})$ is an applicative term.

The \emph{SK Combinatory Calculus} (or just SK-calculus)
is the rewrite system involving terms built from just two atomic combinators, $S$ and $K$:
\[
\begin{array}{ccc}
S\ f\ g\ x & \rightarrow & f\ x\ (g\ x) \\
K\ x\ y & \rightarrow & x
\end{array}
\]

A combinator can also be viewed as a function acting on terms;
hence applicative terms $t$ built from combinators are also functions.
Using just $S$ and $K$,
it is possible to express all functions encodable in the $\lambda$-calculus.
For example, $S\ K\ K$ encodes the identity function.
Figure~\ref{fig:sk-tree} shows the rewrite rules
and the evaluation of the identity
with terms depicted as trees.

From a $\lambda$-calculus perspective,
a combinator can be viewed as a closed $\lambda$-term built by wrapping
a purely applicative term in $\lambda$-abstractions:
\[
\begin{array}{ccc}
S & \equiv & \lambda f . \lambda g . \lambda x . f\ x\ (g\ x) \\
K & \equiv & \lambda x . \lambda y . x
\end{array}
\]
This leads to an obvious translation $\lambdat(t)$ from SK-calculus into $\lambda$-calculus:
\[
\begin{array}{rcl}
\lambdat(S) & \defeq & \lambda f . \lambda g . \lambda x . f\ x\ (g\ x) \\
\lambdat(K) & \defeq & \lambda x . \lambda y . x \\
\lambdat(t_1\ t_2) & \defeq & \lambdat(t_1)\ \lambdat(t_2)
\end{array}
\]
There are a number of translations $\unlambdat(e)$ from $\lambda$-calculus into SK-calculus,
including the following~\cite{Hindley:2008:LCI:1388400}:
\[
\begin{array}{rcll}
\unlambdat(x) & = & x \\
\unlambdat(e_1\ e_2) & = & \unlambdat(e_1)\ \unlambdat(e_2) \\
\unlambdat(\lambda x . e) & = & \unlambdat_x (e) \\
\unlambdat_x(x) & = & S\ K\ K \\
\unlambdat_x(e) & = & K\ \unlambdat(e) & \mbox{if $x$ does not occur free in $e$} \\
\unlambdat_x(e\ x) & = & \unlambdat(e) & \mbox{if $x$ does not occur free in $e$} \\
\unlambdat_x(e_1\ e_2) & = & S\ \unlambdat_x(e_1)\ \unlambdat_x(e_2) & \mbox{if neither of the above applies} \\
\unlambdat_x(\lambda y . e) & = & \unlambdat_x(\unlambdat(\lambda y . e))
\end{array}
\]

This translation is left-inverse to the $\lambda$-translation;
that is $\unlambdat(\lambdat(t)) = t$.
However, it is not right-inverse.

The rewrite rules of combinatory calculus are very simple to implement, as:
there is no need to track bound variables;
the number of rewrite rules is small and fixed;
and all transformations are \emph{local}.
Here ``local'' means that,
viewing a term as a graph,
each transformation involves only a small, bounded number of edge additions and deletions,
all affecting nodes that are either within a bounded distance of the combinator
or are newly created (with the number of new nodes also being bounded).
Because of this simplicity,
combinators have frequently been considered as a basis for hardware or virtual machines
for executing functional programs~\cite{turner1979new,Clarke:1980:SSK:800087.802798}.
Combinators can be thought of as an assembly language for functional programs
(although often an expanded set of combinators~\cite{turner1979another} is used to
avoid a combinatorial explosion in the size of the compiled program).

\subsection{SF Combinatory Calculus}

\begin{figure}[p]
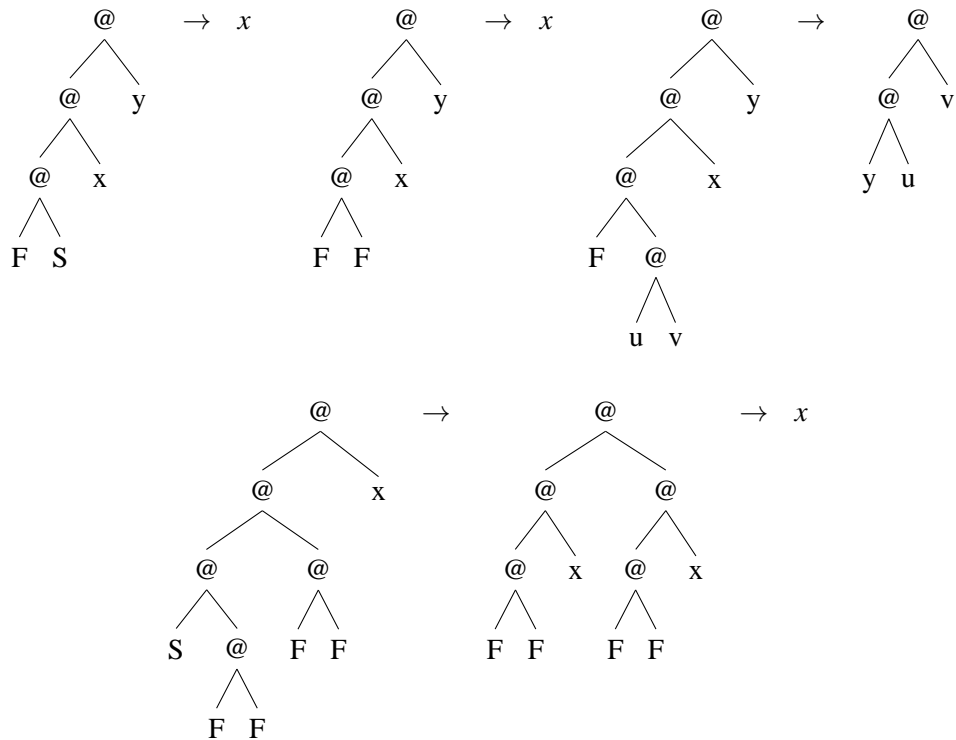

\[
\begin{array}{ccccccccccc}
\Tree [.@ [.@ [.@ F S ] x ] y ] &
\rightarrow &
x &
&  
\Tree [.@ [.@ [.@ F F ] x ] y ] &
\rightarrow &
x 
&
\Tree [.@ [.@ [.@ F [.@ u v ] ] x ] y ] &
\rightarrow &
\Tree [.@ [.@ y u ] v ]
\end{array}
\]
\vspace{1ex}
\[
\begin{array}{ccccc}
\Tree [.@ [.@ [.@ S [.@ F F ] ] [.@ F F ] ] x ] & \rightarrow & \Tree [.@ [.@ [.@ F F ] x ] [.@ [.@ F F ] x ] ]
& \rightarrow & x
\end{array}
\]

\caption{Terms of SF-calculus treated as trees.
Above: the reduction rules for $F$ on atoms and compound terms.
Below: evaluation of the identity function $S\ (F\ F)\ (F\ F)$.}
\label{fig:sf-tree}
\end{figure}

The \emph{SF Combinatory Calculus} is a recently-developed system of combinators
for expressing computation that manipulates
the internal structure of programs~\cite{DBLP:journals/jsyml/Given-WilsonJ11}.
It consists of just two combinators: $S$ and $F$.
$S$ is the same as in SK-calculus.
$F$ is a \emph{factorisation} combinator that allows non-atomic expressions
to be split up into their component parts;
it has two reduction rules (also depicted in Figure~\ref{fig:sf-tree}):

\[
\begin{array}{cccl}
F\ f\ x\ y & \rightarrow & x & \mbox{if $f = S$ or $f = F$} \\
F\ (u\ v)\ x\ y & \rightarrow & y\ u\ v & \mbox{if $u\ v$ is a factorable form}
\end{array}
\]

A \emph{factorable form} is a term of the form
$S$, $S\ u$, $S\ u\ v$, $F$, $F\ u$ or $F\ u\ v$,
for any terms $u$ and $v$;
that is, a term cannot be factorised if it could be reduced at the outermost level.
This ensures that reduction is globally \emph{confluent}, regardless of the reduction order chosen.
It also means that the usual notion of (weak) head reduction is not sufficient for
evaluating programs in this system:
if a term is of the form $F\ f\ x\ y$, then $f$ must be head reduced (if possible)
before applying the reduction rule for $F$.

$F$ stretches our usual notion of what constitutes a combinator slightly,
as it has two rewrite rules,
with the conclusion of the second not being built from application of its arguments,
as it deconstructs the application $u\ v$.
Nonetheless, it is still fair to call SF-calculus a combinatory calculus,
as terms in the calculus are still built solely from application of its atoms
$S$ and $F$.

\paragraph{Confluence and the theory of weak equality.}
Confluence means that,
for any terms $u$, $v$ and $v'$,
if $u \rightarrow^* v$ and $u \rightarrow^* v'$,
it follows that there is a term $w$ with $v \rightarrow^* w$ and $v' \rightarrow^* w$.
This property can be proved for SF-calculus using the standard technique of parallel reductions.
The \emph{weak equality} relation $=_w$ is the symmetric, reflexive, transitive closure
of the reduction relation $\rightarrow$.
From confluence and the fact that the terms $S$ and $F$ are irreducible,
we can conclude that there are terms $u$ and $v$ such that $u \neq_w v$.
That is, the equational theory of $=_w$ for SF-calculus is \emph{consistent}.

The obvious way of adding a factorisation operator to $\lambda$-calculus has
no restriction to factorable forms equivalent to that for $F$.
Consequently, adding this operator breaks confluence,
so the resulting theory of weak equality is not consistent.

\paragraph{Extensional equality.}
Two terms are \emph{extensionally equal}
if they compute the same function, perhaps in different ways.
Within the SF-calculus, it is possible to \emph{distinguish} between two such terms.
Consequently, SF-calculus cannot be translated into $\lambda$-calculus.
For example, consider $t_1 = F\ F\ S$ and $t_2 = F\ S\ S$.
For any term $u$,
we have $t_1\ u = F\ F\ S\ u \rightarrow^* S$
and $t_2\ u = F\ S\ S\ u \rightarrow^* S$,
so $t_1$ and $t_2$ are extensionally equal
(and behave like the term $K\ S$ of SK-calculus).
In SK-calculus or $\lambda$-calculus,
if two terms $t_1$ and $t_2$ are extensionally equal,
then we can replace one with the other without changing the result of a computation.
However, this is not the case in SF-calculus, as we can use $F$ to construct a term $v$
(schematically $v = \lambda t . F\ t\ \_\ (\lambda u . \lambda v . F\ u\ \_\ (\lambda x . \lambda y . y))$) 
such that $v\ t_1 \rightarrow^* F$ and $v\ t_2 \rightarrow^* S$.

In SK-calculus, it is possible to extend the theory of weak equality $=_w$
with a rule corresponding to $\eta$-reduction,
yielding a theory of extensional equality $=_\mathit{ext}$
such that $t_1 =_\mathit{ext} t_2$ if and only if
$t_1$ and $t_2$ are extensionally equal~\cite{Hindley:2008:LCI:1388400}.
Clearly, any reasonable attempt to extend the theory of weak equality for SF-calculus
to an extensional theory of equality will be inconsistent,
as it will equate $S$ with $F$.
This is in direct and deliberate contrast to SK-calculus.

%

\paragraph{Expressivity of SF-calculus.}
There is a translation from SK-calculus into SF-calculus:
$K$ can be expressed as $F\ F$.
Hence all functions expressible in SK-calculus and thus $\lambda$-calculus
are expressible in SF-calculus.

SF-calculus is \emph{structure complete},
in the sense that it can pattern match over normal forms of terms (those having no redexes)
and distinguish between any two different terms in normal form.
In particular, for any two such terms $t_1$ and $t_2$, there is a term $e$ such that
we have $e\ t_1 \rightarrow^* S$ and $e\ t_2 \rightarrow^* F$.
Adding System F types to SF-calculus (and giving names to some other combinators),
the resulting calculus can encode and type
an interpreter for its own language~\cite{DBLP:conf/icfp/JayP11}.
Thus it presents a promising theoretical foundation for reasoning about programs
that transform other programs,
for example by means of partial evaluation.

As a more concrete example of the sorts of programs we might write in SF-calculus,
suppose we have an expression $f\ x\ y$ and
we want to flip its arguments to give $f\ y\ x$~\cite{DBLP:journals/jfp/GrundyMO06}.
For example, perhaps we are writing an optimising compiler,
$f$ is a commutative function that is not strict in both arguments and
we expect $f\ y\ x$ to execute faster than $f\ x\ y$.
Schematically, we could write a program performing this transformation as:
\[
\lambda a . F\ a\ \_\ (\lambda b . \lambda y . F\ b\ \_\ (\lambda f . \lambda x . f\ y\ x ))
\]
where $\_$ is any dummy value. Expressed purely in terms of $S$ and $F$, this can be written as:
\[
(S F (F F S)) (F F (S (S (F F (S (F F S) (F F))) (S F (F F S))) (F F (S (F F (S (S (F F) (F F)))) (F F)))))
\]
Obviously, because of its lack of readability,
SF-calculus (like SK-calculus and $\lambda$-calculus)
is not suitable for use directly by human programmers.

%
%
%
%
%
%

\section{0CFA for SK-Calculus}
Before we can formulate 0CFA for SF-calculus,
we must first consider what it means for SK-calculus.
A central idea in 0CFA for $\lambda$-calculus is that
the analysis computes an over-approximation of
the expressions that may be bound to a variable.
It seems a little perverse to apply this to SK-calculus,
where there are deliberately no variables.

As SK-calculus can be translated into $\lambda$-calculus,
it is easy enough to translate a term $t$ of SK-calculus
into an equivalent $\lambda$-expression $\lambdat(t) = e$ and analyse that.
We could define our analysis by
$\Gamma \models_{\mathit{SK}} t \iff \Gamma \models_{\lambda} \lambdat(t)$.
Furthermore, any SK-calculus reduction $t \rightarrow t'$
corresponds to a sequence of 2 (for $K$) or 3 (for $S$) $\lambda$-calculus $\beta$-reductions.
So for $\lambdat(t') = e'$ we have $e \rightarrow^* e'$.
Then, as 0CFA is coherent with evaluation following an arbitrary $\beta$-reduction strategy,
we would have $\Gamma \models_{\lambda} e'$ and hence $\Gamma \models_{\mathit{SK}} t'$,
showing the coherence of our combinatory 0CFA with evaluation.

But what then is the meaning of the resulting analysis?
We can answer this by reversing the translation
(for example, using $\unlambdat$)
to produce a labelled semantics for SK-calculus
and 0CFA rules that apply directly to SK combinatory terms.

\subsection{Labelled Semantics}

\begin{figure}[h]
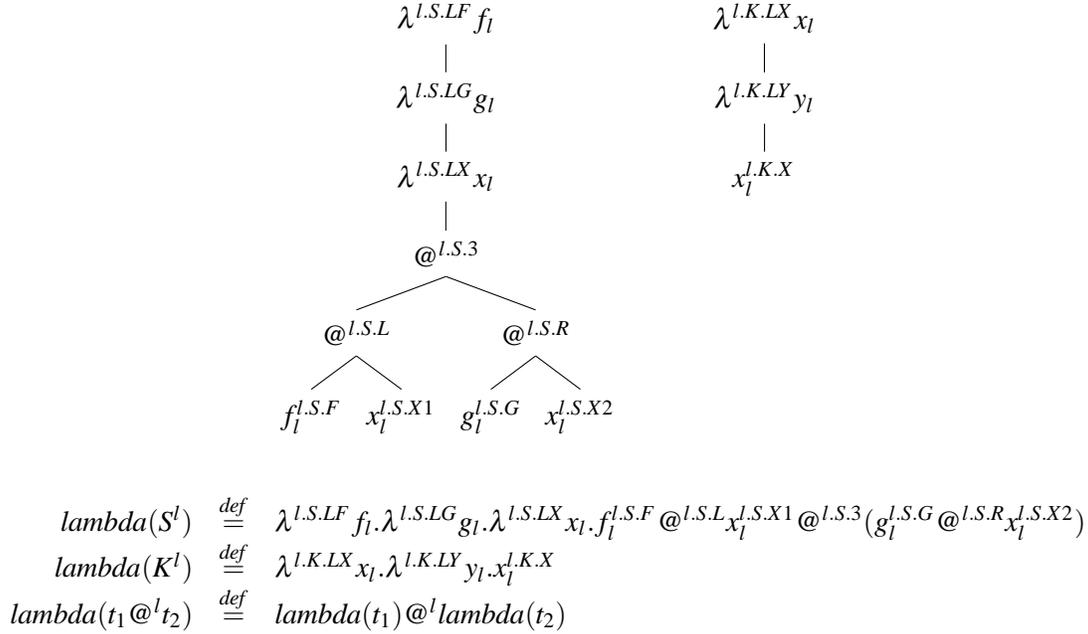

\[
\begin{array}{ccc}
\Tree [.{$\lambda^{l.S.LF} f_l$} [.{$\lambda^{l.S.LG} g_l$} [.{$\lambda^{l.S.LX} x_l$}
[.{$@^{l.S.3}$} [.{$@^{l.S.L}$} $f_l^{l.S.F}$ $x_l^{l.S.X1}$ ] [.{$@^{l.S.R}$} $g_l^{l.S.G}$ $x_l^{l.S.X2}$ ] ] ] ] ]
&
\hspace{2ex}
&
\Tree [.{$\lambda^{l.K.LX} x_l$} [.{$\lambda^{l.K.LY} y_l$} $x_l^{l.K.X}$ ] ]
\end{array}
\]

\[
\begin{array}{rcl}
\lambdat(S^l) & \defeq & \lambda^{l.S.LF} f_l . \lambda^{l.S.LG} g_l . \lambda^{l.S.LX} x_l .
    f_l^{l.S.F} @^{l.S.L} x_l^{l.S.X1} @^{l.S.3} (g_l^{l.S.G} @^{l.S.R} x_l^{l.S.X2}) \\
\lambdat(K^l) & \defeq & \lambda^{l.K.LX} x_l . \lambda^{l.K.LY} y_l . x_l^{l.K.X} \\
\lambdat(t_1 @^l t_2) & \defeq & \lambdat(t_1) @^l \lambdat(t_2)
\end{array}
\]
\caption{The labelled $\lambda$-calculus translation $\lambdat(t)$ (below),
with $\lambdat(S^l)$ (upper-left) and $\lambdat(K^l)$ (upper-right) illustrated as trees.}
\label{fig:label-lambda}
\end{figure}

First we will look at the result of $\beta$-reducing
expressions $\lambdat(S\ f\ g\ x)$ and $\lambdat(K\ x\ y)$
using the labelled semantics of $\lambda$-calculus.
We begin by extending $\lambdat$ to produce labelled expressions,
as shown in Figure~\ref{fig:label-lambda}.

Here we have extended $\Label$ to give an easy, syntactic way
of associating a fixed set of ``sublabels'' with each ordinary label.
The choice of names for the sublabels is somewhat arbitrary,
although we have chosen them to match the structure of the expressions.
It is important at this stage that the label on each expression remains distinct,
so that we do not lose precision in formulating our analysis.
We can now use $\lambdat$ to produce labelled reduction rules for SK-calculus:
\[
\begin{array}{rcl}
K^{l_2} @^{l_3} x^{l_1} @^{l_4} y^{l_0} & \rightarrow & x^{l_1} \\
S^{l_3} @^{l_4} f^{l_2} @^{l_5} g^{l_1} @^{l_6} x^{l_0} & \rightarrow &
(f^{l_2} @^{l_3.S.L} x^{l_0}) @^{l_3.S.3} (g^{l_1} @^{l_3.S.R} x^{l_0})
\end{array}
\]

Note that in the conclusion of the reduction of $S$,
there are new labels that were not present in the original program.
These are the sublabels from the applications introduced by $\lambdat$.
In the $\lambda$-calculus formulation,
these are present inside the $\lambda$-expression before reduction;
the reduction exposes them.
A consequence of this is that, in analysing a term of SK-calculus,
we will need to consider labels that do not occur in the term.
If the set of labels were infinite, this might pose a problem for an analysis.
However, this is not the case,
as the names of the labels are syntactically derived from the label on $S$;
only a finite, statically derivable set of labels may arise during the execution of a term.

\subsection{Analysis Rules}

\begin{figure}

\[
\begin{array}{lrcccl}
\mbox{Base Labels} & \mathbb{N} & \ni & n & & \\
\mbox{Sublabel Names} & & & s & ::= & S.0 \mid S.1 \mid S.2 \mid S.3 \mid S.L \mid S.R \mid K.0 \\
\mbox{Labels} & \Label & \ni & l & ::= & n \mid n.s \\
\mbox{Labelled Terms} & & & t & ::= & S^n \mid K^n \mid t_1 @^l t_2 \mid \ang{x}^l \\
\mbox{Abstract Values} & \Abs & \ni & v & ::= & S_0^{n} \mid S_1^{n} \mid S_2^{n} \mid K_0^{n} \mid K_1^{n} \\
\mbox{Abstract Environment} & & & \Gamma & : & \Label \rightarrow \mathcal{P}(\Abs) \\
\mbox{Abstract Activation} & & & \varphi & : & \Label \rightarrow \Bool
\end{array}
\]
\[
\begin{array}{lcl}
\Gamma, \varphi \models S^n & \iff & S_0^{n} \in \Gamma(n) \wedge (\varphi(n) \Rightarrow \Gamma,\varphi \models t_{S^n}) \\
\Gamma, \varphi \models K^n & \iff & K_0^{n} \in \Gamma(n) \\
\Gamma, \varphi \models t_1^{l_1} @^{l_3} t_2^{l_2} & \iff &
    \phantom{\mbox{} \wedge \mbox{}} \Gamma, \varphi \models t_1 \wedge \Gamma, \varphi \models t_2 \\
& & \mbox{} \wedge \forall S_0^n \in \Gamma(l_1) . \Gamma(l_2) \subseteq \Gamma(n.S.0) \wedge S_1^n \in \Gamma(l_3) \\
& & \mbox{} \wedge \forall S_1^n \in \Gamma(l_1) . \Gamma(l_2) \subseteq \Gamma(n.S.1) \wedge S_2^n \in \Gamma(l_3) \\
& & \mbox{} \wedge \forall S_2^n \in \Gamma(l_1) . \Gamma(l_2) \subseteq \Gamma(n.S.2) \wedge \Gamma(n.S.3) \subseteq \Gamma(l_3)
    \wedge \varphi(n) \\
& & \mbox{} \wedge \forall K_0^n \in \Gamma(l_1) . \Gamma(l_2) \subseteq \Gamma(n.K.0) \wedge K_1^n \in \Gamma(l_3) \\
& & \mbox{} \wedge \forall K_1^n \in \Gamma(l_1) . \Gamma(n.K.0) \subseteq \Gamma(l_3) \\
\Gamma, \varphi \models \ang{x}^l & \iff & \mathit{true}
\end{array}
\]
\[
t_{S^n} \defeq (\ang{f}^{n.S.0} @^{n.S.L} \ang{x}^{n.S.2}) @^{n.S.3} (\ang{g}^{n.S.1} @^{n.S.R} \ang{x}^{n.S.2})
\]

\caption{0CFA for SK-calculus}
\label{fig:0cfa-sk}
\end{figure}

We are now able to translate the rules of 0CFA for $\lambda$-calculus
into new rules for SK-calculus, as shown in Figure~\ref{fig:0cfa-sk}.
Note that a label is now either a base label (as before, taken from $\mathbb{N}$)
or a base label suffixed with a sublabel name (taken from a fixed, finite set).
In performing the translation,
we have eliminated some unnecessary or trivial constraints,
such as those for tracking the 2nd argument to $K$, which is never used.
We have restricted the grammar of abstract values
to just instances of $S$ and $K$ with different numbers of arguments applied.
We have also made a small change from 0CFA for $\lambda$-calculus.
The constraints that express the result of reducing an $S$
are only \emph{activated} if it is possible for that $S$ to be reduced.
This may improve precision slightly,
but would be unsound in the $\lambda$-calculus setting,
where we can reduce the expression corresponding to $S$
even if it only has 1 or 2 arguments.
It will be more important for SF-calculus.
In order to track whether constraints for an instance have been activated,
we introduce a new component $\varphi : \Label \rightarrow \Bool$ to the solution of the constraints,
with $\varphi(n)$ being true when the constraints for $S^n$ are active.

The intuitive meaning of $S_0^n \in \Gamma(l)$ is that $S^n$ may occur at the point labelled $l$,
hence the rule for $\Gamma, \varphi \models S^n$.
The meaning of $S_1^n \in \Gamma(l)$ is that a term built from applying 1 argument to $S^n$ may occur at $l$,
or that $S^n$ may occur as the 1st left child of the term tree node labelled $l$.
The meaning of $S_2^n$ is analogous (but for 2 arguments or the 2nd left child),
as is that of $K_0^n$ and $K_1^n$ (but for $K$, not $S$).

The abstract values in $\Gamma(n.S.0)$ are meant to over-approximate
the values that may occur as the 1st argument to $S^n$;
similarly for $\Gamma(n.S.1)$ and the 2nd argument,
and analogously for $n.S.2$ and $n.K.0$.

This leads to the explanation of the conjunction of conditions for $\Gamma, \varphi \models t_1^{l_1} @^{l_3} t_2^{l_2}$.
For example, the condition involving $\forall S_0^n$ ensures that, if $S^n$ may occur in function position, then:
the abstraction $\Gamma(n.S.0)$ of the 1st argument of $S^n$ over-approximates
the arguments that may be supplied by $t_2$;
and the result of the application needs only 2 more arguments for a reduction to occur.
The condition involving $\forall S_1^n$ and the first part of the condition with $\forall S_2^n$ are similar.
The condition on $\forall K_0^n$ is analogous to that for $\forall S_0^n$.
The condition with $\forall K_1^n$ simply says that the result of reducing $K$
may be anything that occurs as its 1st argument.

The second part of the condition on $\forall S_2^n$ is more complicated.
In the event that $S^n$ may receive 3 arguments and hence be reduced,
it introduces constraints for the conclusion of the reduction,
which are those generated by analysis of the constant applicative term $t_{S^n}$.
It also says that the result of the reduction may be anything that occurs at the root of that term,
which has label $n.S.3$.

The introduction of the constraints for $t_{S^n}$ is forced by asserting $\varphi(n)$,
which produces the corresponding constraints in the rule for $\Gamma, \varphi \models S^n$.
The use of $\varphi$ avoids the introduction of a recursive loop in the constraint rules;
an alternative method would be to use a coinductive definition of $\models$.
Within $t_{S^n}$ we use dummy terms of the form $\ang{x}^l$ to give the leaf node of the term tree a label;
here $f$, $g$ and $x$ have no meaning (other than to make reading the rules easier) and play no role in the analysis.


This analysis may seem like a step backwards,
as we have replaced a small set of general rules for $\lambda$-calculus
with a larger, more specific set of rules for SK-calculus.
However, there are a number of benefits.
Firstly, the rules for S can be used directly in 0CFA for SF-calculus.
Secondly, they reveal the meaning of 0CFA in SK-calculus:
the abstract values at a labelled point tell us which combinators may occur
at that point \emph{and} locally at its left branches.
This insight will be key in both producing an accurate analysis for F
and for justifying why it is reasonable to call that analysis 0CFA.
Finally, because SK-calculus does not have to deal
with arbitrary substitution or the intricacies of name-binding,
the proof of correctness for this system is considerably simpler
than that for $\lambda$-calculus.

Recall the example of 0CFA for $\lambda$-calculus involving
applying the identity function to itself.
The corresponding SK-calculus term and a solution for its analysis are shown
in Figure~\ref{fig:sk-example-analysis}.
Note that $\Gamma(10) = \Gamma(4) = \{ S^2_2 \}$,
indicating that the result of evaluation has $S^2$ as its second left child;
that is, it is the second identity function $(S^2 @^3 K^1 @^4 K^0)$.

\begin{figure}
\[
\begin{array}{rclcrclcrclcrcl}
\Gamma(0) & = & \{ K^0_0 \} & &
\Gamma(1) & = & \{ K^1_0 \} & &
\Gamma(2) & = & \{ S^2_0 \} & &
\Gamma(2.S.0) & = & \{ K^1_0 \} \\
\Gamma(2.S.1) & = & \{ K^0_0 \} & &
\Gamma(3) & = & \{ S^2_1 \} & &
\Gamma(4) & = & \{ S^2_2 \} & &
\Gamma(5) & = & \{ K^5_0 \} \\
\Gamma(5.K.0) & = & \{ S^2_2 \} & &
\Gamma(6) & = & \{ K^6_0 \} & &
\Gamma(6.K.0) & = & \{ S^2_2 \} & &
\Gamma(7) & = & \{ S^7_0 \} \\
\Gamma(7.S.0) & = & \{ K^6_0 \} & &
\Gamma(7.S.1) & = & \{ K^5_0 \} & &
\Gamma(7.S.2) & = & \{ S^2_2 \} & &
\Gamma(7.S.3) & = & \{ S^2_2 \} \\
\Gamma(7.S.L) & = & \{ K^6_1 \} & &
\Gamma(7.S.R) & = & \{ K^5_1 \} & &
\Gamma(8) & = & \{ S^7_1 \} & &
\Gamma(9) & = & \{ S^7_2 \} \\
\Gamma(10) & = & \{ S^2_2 \} & &
\varphi(7) & = & \mathit{true}
\end{array}
\]
\[
\Gamma, \varphi \models (S^7 @^8 K^6 @^9 K^5) @^{10} (S^2 @^3 K^1 @^4 K^0)
\]
\caption{Solution of the analysis for application of identity to itself in SK-calculus.}
\label{fig:sk-example-analysis}
\end{figure}

\subsection{Correctness}

We now prove the correctness of this analysis for SK-calculus.
First we make some observations about the satisfaction of constraints:

\begin{lemma}[SK Substitution]
If $\Gamma, \varphi \models t_1^{l_1} @^{l_3} t_2^{l_2}$,
as well as $\Gamma, \varphi \models {t'_1}^{l'_1}$ and $\Gamma, \varphi \models {t'_2}^{l'_2}$,
with $\Gamma(l'_1) \subseteq \Gamma(l_1)$,
$\Gamma(l'_2) \subseteq \Gamma(l_2)$
and $\Gamma(l'_3) \supseteq \Gamma(l_3)$
then $\Gamma, \varphi \models {t'_1}^{l'_1} @^{l'_3} {t'_2}^{l'_2}$.
\end{lemma}
\begin{proof}
Trivial by inspection of the constraints generated by $@$.
\end{proof}

\begin{lemma}[SK Reduction Coherence]
For any top-level reduction $t^l \rightarrow {t'}^{l'}$,
if $\Gamma, \varphi \models t^l$ then
\emph{1)} $\Gamma, \varphi \models {t'}^{l'}$ and
\emph{2)} $\Gamma(l') \subseteq \Gamma(l)$.
\end{lemma}
\begin{proof}
Case split on the two kinds of top-level reduction ($S$ and $K$).

\emph{Case $S$:}
We have $t^l = S^{l_3} @^{l_4} f^{l_2} @^{l_5} g^{l_1} @^{l_6} x^{l_0}$
and $t'^{l'} = (f^{l_2} @^{l_3.S.L} x^{l_0}) @^{l_3.S.3} (g^{l_1} @^{l_3.S.R} x^{l_0})$.
Expanding the constraints for $\Gamma, \varphi \models t^l$, we have:
$S^{l_3}_0 \in \Gamma(l_3)$;
$S^{l_3}_1 \in \Gamma(l_4)$;
$S^{l_3}_2 \in \Gamma(l_5)$; hence
$\Gamma(l_3.S.3) \subseteq \Gamma(l_6)$ (proving Condition 2) and $\varphi(l_3)$ is true.
As $\Gamma, \varphi \models S^{l_3}_0$ and $\varphi(l_3)$, we have $\Gamma, \varphi \models t_{S_{l_3}}$.
But $t_{S_{l_3}}$ can be turned into $t'$ by substituting $f$, $g$ and $x$ at its leaves.
So to prove Condition 1, we just need to show that we can use the Substitution Lemma.
Now as $\Gamma, \varphi \models t^l$, we get:
$\Gamma, \varphi \models f^{l_2}$;
$\Gamma, \varphi \models g^{l_1}$; and
$\Gamma, \varphi \models x^{l_0}$.
Furthermore:
from $S^{l_3}_0 \in \Gamma(l_3)$ we get $\Gamma(l_2) \subseteq \Gamma(l_3.S.0)$;
from $S^{l_3}_1 \in \Gamma(l_4)$ we get $\Gamma(l_1) \subseteq \Gamma(l_3.S.1)$; and
from $S^{l_3}_2 \in \Gamma(l_3)$ we get $\Gamma(l_0) \subseteq \Gamma(l_3.S.2)$.
So the Substitution Lemma can be used to prove Condition 1.

\emph{Case $K$:}
We have $t^l = K^{l_2} @^{l_3} x^{l_1} @^{l_4} y^{l_0}$
and $t^{l'} = x^{l_1}$.
From $\Gamma, \varphi \models t^l$ we have $\Gamma, \varphi \models x^{l_1}$,
showing Condition 1.
Expanding the constraints for $\Gamma, \varphi \models t^l$ further, we get:
$K^{l_2}_0 \in \Gamma(l_2)$;
hence $\Gamma(l_1) \subseteq \Gamma(l_2.K.0)$ and $K^{l_2}_1 \in \Gamma(l_3)$;
thus $\Gamma(l_2.K.0) \subseteq \Gamma(l_4)$.
Combining these gives $\Gamma(l_1) \subseteq \Gamma(l_4)$, proving Condition 2.
\end{proof}

\begin{theorem}[SK Evaluation Coherence]
For any reduction in context $C[t^{l_1}]^{l_2} \rightarrow C[t'^{l'_1}]^{l'_2}$,
if $\Gamma, \varphi \models C[t^{l_1}]^{l_2}$ then
\emph{1)} $\Gamma, \varphi \models C[t'^{l'_1}]^{l'_2}$ and
\emph{2)} $\Gamma(l'_2) \subseteq \Gamma(l_2)$.
\end{theorem}
\begin{proof}
For the empty context, this follows immediately from the Reduction Coherence Lemma.
For a non-empty context $C$, Condition 2 is trivially true as $l'_2 = l_2$.
For Condition 1, the reduction occurs at either the left child or right child
of an application node in the term tree (as all other nodes are leaves).
Any constraints generated by the context are unchanged and hence remain satisfied.
For the hole in the context, from $\Gamma, \varphi \models C[t^{l_1}]^{l_2}$
we have $\Gamma, \varphi \models t^{l_1}$,
so by Reduction Coherence we get $\Gamma, \varphi \models {t'}^{l'_1}$
with $\Gamma(l'_1) \subseteq \Gamma(l_1)$.
Hence any constraints within ${t'}^{l'_1}$ are satisfied.
That just leaves the constraints generated by the interaction
between the application at the hole of the context and $t'$.
We can apply the Substitution Lemma to the application node to show that they are satisfied,
which gives Condition 1 as required.
\end{proof}

\begin{corollary}[SK Soundness]
If $\Gamma, \varphi \models t^l$ and $t \rightarrow^* t'$ then $\Gamma, \varphi \models {t'}^{l'}$.
\end{corollary}
\begin{proof}
By induction over the length of the derivation of $\rightarrow^*$
and application of the Evaluation Coherence Theorem.
\end{proof}

\section{0CFA for SF-Calculus}
We now turn our attention to formulation of 0CFA for SF-calculus.
As $F$ is not encodable in $\lambda$-calculus,
we cannot argue for the correctness of our analysis by appeal to the translation $\lambdat$.
Instead, we must follow the style of our formulation for SK-calculus.

\subsection{Labelled Semantics}

Following the labelled reduction for $S$,
we introduce the following labelled reductions for $F$:
\[
\begin{array}{rcll}
F^{l_3} @^{l_4} f^{l_2} @^{l_5} x^{l_1} @^{l_6} y^{l_0} & \rightarrow &
x^{l_1} & \mbox{if $f=S$ or $f=K$} \\
F^{l_3} @^{l_4} (u^{l_7} @^{l_2} v^{l_8}) @^{l_5} x^{l_1} @^{l_6} y^{l_0} & \rightarrow &
(y^{l_0} @^{l_3.F.M} u^{l_7}) @^{l_3.F.3} v^{l_8}
 & \mbox{if $u\ v$ is a factorable form}
\end{array}
\]

\subsection{Analysis Rules}

\begin{figure}[h]

\[
\begin{array}{lrcccl}
\mbox{Base Labels} & \mathbb{N} & \ni & n & & \\
\mbox{Sublabel Names} & & & s & ::= & S.0 \mid S.1 \mid S.2 \mid S.3 \mid S.L \mid S.R \mid \\
 & & & & &                            F.0 \mid F.1 \mid F.2 \mid F.3 \mid F.L \mid F.R \mid F.M \\
\mbox{Labels} & \Label & \ni & l & ::= & n \mid n.s \\
\mbox{Labelled Terms} & & & t & ::= & S^n \mid F^n \mid t_1 @^l t_2 \mid \ang{x}^l \\
\mbox{Abstract Values} & \Abs & \ni & v & ::= & S_0^{n} \mid S_1^{n} \mid S_2^{n} \mid
                                                F_0^{n} \mid F_1^{n} \mid F_2^{n} \mid @^{(l_1, l_2)} \\
\mbox{Abstract Environment} & & & \Gamma & : & \Label \rightarrow \mathcal{P}(\Abs) \\
\mbox{Abstract Activation} & & & \varphi & : & \Label \rightarrow \Bool
\end{array}
\]
\[
\begin{array}{lcl}
\Gamma, \varphi \models S^n & \iff & S_0^{n} \in \Gamma(n) \wedge (\varphi(n) \Rightarrow \Gamma,\varphi \models t_{S^n}) \\
\Gamma, \varphi \models F^n & \iff & \phantom{\mbox{} \wedge \mbox{}} F_0^{n} \in \Gamma(n) \\
& & \mbox{} \wedge \varphi(n) \Rightarrow (\exists n_0 . S_0^{n_0} \in \Gamma(n.F.0) \vee F_0^{n_0} \in \Gamma(n.F.0))
               \Rightarrow \Gamma(n.F.1) \subseteq \Gamma(n.F.3) \\
& & \mbox{} \wedge \varphi(n) \Rightarrow (\exists n_0 . S_1^{n_0} \in \Gamma(n.F.0) \vee S_2^{n_0} \in \Gamma(n.F.0) \vee \mbox{} \\
& & \phantom{\mbox{} \wedge \varphi(n) \Rightarrow \exists n_0 .} F_1^{n_0} \in \Gamma(n.F.0) \vee F_2^{n_0} \in \Gamma(n.F.0) )
\Rightarrow \Gamma,\varphi \models t_{F^n} \wedge \mbox{} \\
& & \phantom{\mbox{} \wedge \varphi(n) \Rightarrow \mbox{}}
\forall @^{l_1,l_2} \in \Gamma(n.F.0) . \Gamma(l_1) \subseteq \Gamma(n.F.L) \wedge \Gamma(l_2) \subseteq \Gamma(n.F.R) \\
\Gamma, \varphi \models t_1^{l_1} @^{l_3} t_2^{l_2} & \iff &
    \phantom{\mbox{} \wedge \mbox{}} \Gamma, \varphi \models t_1 \wedge \Gamma, \varphi \models t_2 \\
& & \mbox{} \wedge \exists @^{l_4, l_5} \in \Gamma(l_3) . \Gamma(l_1) \subseteq \Gamma(l_4) \wedge \Gamma(l_2) \subseteq \Gamma(l_5) \\
& & \mbox{} \wedge \forall S_0^n \in \Gamma(l_1) . \Gamma(l_2) \subseteq \Gamma(n.S.0) \wedge S_1^n \in \Gamma(l_3) \\
& & \mbox{} \wedge \forall S_1^n \in \Gamma(l_1) . \Gamma(l_2) \subseteq \Gamma(n.S.1) \wedge S_2^n \in \Gamma(l_3) \\
& & \mbox{} \wedge \forall S_2^n \in \Gamma(l_1) . \Gamma(l_2) \subseteq \Gamma(n.S.2) \wedge \Gamma(n.S.3) \subseteq \Gamma(l_3)
    \wedge \varphi(n) \\
& & \mbox{} \wedge \forall F_0^n \in \Gamma(l_1) . \Gamma(l_2) \subseteq \Gamma(n.F.0) \wedge F_1^n \in \Gamma(l_3) \\
& & \mbox{} \wedge \forall F_1^n \in \Gamma(l_1) . \Gamma(l_2) \subseteq \Gamma(n.F.1) \wedge F_2^n \in \Gamma(l_3) \\
& & \mbox{} \wedge \forall F_2^n \in \Gamma(l_1) . \Gamma(l_2) \subseteq \Gamma(n.F.2) \wedge \Gamma(n.F.3) \subseteq \Gamma(l_3)
    \wedge \varphi(n) \\
\Gamma, \varphi \models \ang{x}^l & \iff & \mathit{true}
\end{array}
\]
\[
\begin{array}{rcl}
t_{S^n} & \defeq & (\ang{f}^{n.S.0} @^{n.S.L} \ang{x}^{n.S.2}) @^{n.S.3} (\ang{g}^{n.S.1} @^{n.S.R} \ang{x}^{n.S.2}) \\
t_{F^n} & \defeq & (\ang{y}^{n.F.2} @^{n.F.M} \ang{u}^{n.F.L}) @^{n.F.3} \ang{v}^{n.F.R}
\end{array}
\]

\caption{0CFA for SF-calculus}
\label{fig:0cfa-sf}
\end{figure}

There are two main problems to consider in analysing $F$:
how to determine whether the 1st argument is a factorable form
and, when that argument is a factorable form,
how to deconstruct its abstract representation.

Concerning the first problem,
if we think back to our analysis for SK-calculus,
a term $t^l$ might evaluate to an atom $S^n$ or $F^n$ if (for some $n$)
$S^n_0 \in \Gamma(l)$ or $K^n_0 \in \Gamma(l)$.
Its normal form might be a non-atomic term
if $\Gamma(l)$ contains any other abstract values.
We can use the same idea for SF-calculus,
except with $F^n_0$ in place of $K^n_0$.

As for the second problem,
in order to deconstruct abstract values,
we introduce a new type of abstract value $@^{l_1,l_2}$.
Intuitively, the abstract value indicates that any concrete value was produced by
applying a term approximated by $\Gamma(l_1)$
to a term approximated by $\Gamma(l_2)$.
The resulting analysis is shown in Figure~\ref{fig:0cfa-sf}.

We have reused the analysis rules for $S$.
The rules for $F$ are mostly very similar.
This is to be expected, as they both take 3 arguments.
In the rule for $\Gamma, \varphi \models F^n$,
there are two separate sets of constraints
that can be activated by $\varphi(n)$,
corresponding to the two reduction rules.
Both involve a further condition that corresponds to
testing whether the 1st argument may be atomic or a factorable form.
The conclusion to the first, corresponding to the atomic case,
is similar to the last rule for $K$ in
$\Gamma, \varphi \models t_1^{l_1} @^{l_3} t_2^{l_2}$
of the analysis for SK-calculus.
The conclusion to the second, which handles factorisation,
introduces the constraints generated by the applicative term $t_{F^n}$
in a similar style to the case for $S$ and $t_{S^n}$.
However, it also adds new constraints to $t_{F^n}$
corresponding to the factorisation of the 1st argument.

There is also a new constraint for $\Gamma, \varphi \models t_1^{l_1} @^{l_3} t_2^{l_2}$ that
introduces abstract values of the form $@^{l_1,l_2}$.
When analysing a term, this is easily satisfied
by setting $@^{l_1,l_2} \in \Gamma(l_3)$.
The slightly more complicated constraint here is necessary
to ensure coherence of the analysis with evaluation.

A term $t^l$ can be analysed by finding a $\Gamma$ and $\varphi$
such that $\Gamma, \varphi \models t^l$.
This is done by solving the constraints using a fixed point process,
much as with 0CFA for $\lambda$-calculus.
We need only consider $\mathcal{O}(n)$ abstract values
(corresponding to nodes in the term tree of $t$),
so we retain the polynomial time complexity of 0CFA.

Consider once again the example of applying the identity function to itself.
The corresponding SF-calculus term and its analysis are shown in
Figure~\ref{fig:sf-example-analysis};
note that $S^6_2 \in \Gamma(18)$, indicating the result correctly.

\begin{figure}
\[
\begin{array}{rclcrcl}
\Gamma(0) & = & \{ F^0_0 \} & &
\Gamma(1) & = & \{ F^1_0 \} \\
\Gamma(1.F.0) & = & \{ F^0_0 \} & &
\Gamma(2) & = & \{ F^1_1, @^{(1,0)} \} \\
\Gamma(3) & = & \{ F^3_0 \} & &
\Gamma(4) & = & \{ F^4_0 \} \\
\Gamma(4.F.0) & = & \{ F^3_0 \} & &
\Gamma(5) & = & \{ F^4_1, @^{(4,3)} \} \\
\Gamma(6) & = & \{ S^6_0 \} & &
\Gamma(6.S.0) & = & \{ F^4_1, @^{(4,3)} \} \\
\Gamma(6.S.1) & = & \{ F^1_1, @^{(1,0)} \} & &
\Gamma(7) & = & \{ S^6_1, @^{(6,5)} \} \\
\Gamma(8) & = & \{ S^6_2, @^{(7,2)} \} & &
\Gamma(9) & = & \{ F^9_0 \} \\
\Gamma(10) & = & \{ F^{10}_0 \} & &
\Gamma(10.F.0) & = & \{ F^9_0 \} \\
\Gamma(10.F.1) & = & \{ S^6_2, @^{(7,2)} \} & &
\Gamma(11) & = & \{ F^{10}_1, @^{(10,9)} \} \\
\Gamma(12) & = & \{ F^{12}_0 \} & &
\Gamma(13) & = & \{ F^{13}_0 \} \\
\Gamma(13.F.0) & = & \{ F^{12}_0 \} & &
\Gamma(13.F.1) & = & \{ S^6_2, @^{(7,2)} \} \\
\Gamma(13.F.3) & = & \{ S^6_2, @^{(7,2)} \} & &
\Gamma(13.F.2) & = & \{ F^{10}_2, @^{(15.S.1,15.S.2)} \} \\
\Gamma(14) & = & \{ F^{13}_1, @^{(13,12)} \} & &
\Gamma(15) & = & \{ S^{15}_0 \} \\
\Gamma(15.S.0) & = & \{ F^{13}_1, @^{(13,12)} \} & &
\Gamma(15.S.1) & = & \{ F^{10}_1, @^{(10,9)} \} \\
\Gamma(15.S.2) & = & \{ S^6_2, @^{(7,2)} \} & &
\Gamma(15.S.3) & = & \{ S^6_2, @^{(15.S.L,15.S.R)}, @^{(7,2)} \} \\
\Gamma(15.S.L) & = & \{ F^{13}_2, @^{(15.S.0,15.S.2)} \} & &
\Gamma(16) & = & \{ S^{15}_1, @^{(15,14)} \} \\
\Gamma(15.S.R) & = & \{ F^{10}_2, @^{(15.S.1,15.S.2)} \} & &
\Gamma(17) & = & \{ S^{15}_2, @^{(16,11)} \} \\
\Gamma(18) & = & \{ \lefteqn{S^6_2, @^{(15.S.L,15.S.R)}, @^{(17,8)}, @^{(7,2)} \} } \\
\varphi(13) & = & \mathit{true} & &
\varphi(15) & = & \mathit{true}
\end{array}
\]
\[
\Gamma, \varphi \models (S^{15} @^{16} (F^{13} @^{14} F^{12}) @^{17} (F^{10} @^{11} F^{9}))
@^{18} (S^{6} @^{7} (F^{4} @^{5} F^{3}) @^{8} (F^{1} @^{2} F^{0}))
\]
\caption{Solution of the analysis for application of identity to itself in SF-calculus.}
\label{fig:sf-example-analysis}
\end{figure}

\subsection{Correctness}

Correctness of the analysis follows by the same sequence of results as for SK-calculus.

\begin{lemma}[SF Substitution]
If $\Gamma, \varphi \models t_1^{l_1} @^{l_3} t_2^{l_2}$,
as well as $\Gamma, \varphi \models {t'_1}^{l'_1}$ and $\Gamma, \varphi \models {t'_2}^{l'_2}$,
with $\Gamma(l'_1) \subseteq \Gamma(l_1)$,
$\Gamma(l'_2) \subseteq \Gamma(l_2)$
and $\Gamma(l'_3) \supseteq \Gamma(l_3)$
then $\Gamma, \varphi \models {t'_1}^{l'_1} @^{l'_3} {t'_2}^{l'_2}$.
\end{lemma}
\begin{proof}
Again, trivial by inspection of the constraints generated by $@$.
It is at this point that the correct formulation of the constraint
$\exists @^{l_4,l_5} \in \Gamma(l_3) . \Gamma(l_1) \subseteq \Gamma(l_4) \wedge \Gamma(l_2) \subseteq \Gamma(l_5)$
is important;
the lemma does not hold if we use the simpler constraint $@^{l_1,l_2} \in \Gamma(l_3)$.
\end{proof}

\begin{lemma}[SF Reduction Coherence]
For any top-level reduction $t^l \rightarrow {t'}^{l'}$,
if $\Gamma, \varphi \models t^l$ then
\emph{1)} $\Gamma, \varphi \models {t'}^{l'}$ and
\emph{2)} $\Gamma(l') \subseteq \Gamma(l)$.
\end{lemma}
\begin{proof}
Case split on the two kinds of top-level reduction ($S$ and $F$).

\emph{Case $S$:} Largely as for SK-calculus.
The only new point is that we must check the constraints on the abstract $@$ values
generated by $@^{l_3.S.L}$ and $@^{l_3.S.R}$ still hold.

\emph{Case $F$:}
We have $t^l = F^{l_3} @^{l_4} f^{l_2} @^{l_5} x^{l_1} @^{l_6} y^{l_0}$.
We begin as for $S$ by expanding the constraints for $\Gamma, \varphi \models t^l$.
We get
$F^{l_3}_0 \in \Gamma(l_3)$,
$F^{l_3}_1 \in \Gamma(l_4)$ and
$F^{l_3}_2 \in \Gamma(l_5)$;
also
$\Gamma(l_2) \subseteq \Gamma(l_3.F.0)$,
$\Gamma(l_1) \subseteq \Gamma(l_3.F.1)$ and
$\Gamma(l_0) \subseteq \Gamma(l_3.F.2)$;
as well as $\Gamma(l_3.F.3) \subseteq \Gamma(l_6)$ and $\varphi(l_3)$.
Now there are two subcases depending on whether $f$ is factorable.

\emph{Subcase $f$ is not factorable:}
We have $t'^{l'} = x^{l_1}$.
From $\Gamma, \varphi \models t^l$
we get $\Gamma, \varphi \models x^{l_1}$, proving Condition 1.
As $f^{l_2}$ is not factorable,
either $f = F^{l_2}$ and $F^{l_2}_0 \in \Gamma(l_2) \subseteq \Gamma(l_3.F.0)$
or $f = S^{l_2}$ and $S^{l_2}_0 \in \Gamma(l_2) \subseteq \Gamma(l_3.F.0)$.
In either case, noting we already have $\Gamma, \varphi \models F^{l_3}$ and $\varphi(l_3)$,
we get $\Gamma(l_3.F.1) \subseteq \Gamma(l_3.F.3)$.
Combining this with $\Gamma(l_1) \subseteq \Gamma(l_3.F.1)$ and $\Gamma(l_3.F.3) \subseteq \Gamma(l_6)$
gives $\Gamma(l_1) \subseteq \Gamma(l_6)$, proving Condition 2.

\emph{Subcase $f$ is factorable:}
We have $f^{l_2} = u^{l_7} @^{l_2} v^{l_8}$
and $t'^{l'} = (y^{l_0} @^{l_3.F.M} u^{l_7}) @^{l_3.F.3} v^{l_8}$.
Condition 2 now follows immediately from $\Gamma(l_3.F.3) \subseteq \Gamma(l_6)$.
As $f$ is factorable, either $u^{l_7}$ or its left child $w^{l_9}$ (if it has one) is $S$ or $F$.
Hence one of $F^{l_7}_1$, $S^{l_7}_1$, $F^{l_9}_2$ and $F^{l_9}_2$ must be in $\Gamma(l_3.F.0)$.
Noting $\Gamma, \varphi \models F^{l_3}$ and $\varphi(l_3)$,
we now have $\Gamma, \varphi \models t_{F^n}$,
as well as a constraint relating abstract $@$ values in $\Gamma(l_3.F.0)$
with $\Gamma(l_3.F.L)$ and $\Gamma(l_3.F.R)$.
Similarly to the case for $S$, in order to prove Condition 1,
we note that we can obtain $t'^{l'}$ by substituting $y^{l_0}$, $u^{l_7}$ and $v^{l_8}$ into $t_{F^n}$,
so we need to show that the Substitution Lemma is applicable.
For $y^{l_0}$ this is easy, as we already have $\Gamma(l_0) \subseteq \Gamma(l_3.F.2)$.
For $u^{l_7}$ and $v^{l_8}$, there must exist some $@^{l_A,l_B} \in \Gamma(l_2)$
with $\Gamma(l_7) \subseteq \Gamma(l_A)$ and $\Gamma(l_8) \subseteq \Gamma(l_B)$.
But $\Gamma(l_2) \subseteq \Gamma(l_3.F.0)$, so $@^{l_A,l_B} \in \Gamma(l_3.F.0)$.
Then, using the above constraint on abstract $@$ values,
$\Gamma(l_A) \subseteq \Gamma(l_3.F.L)$ and $\Gamma(l_B) \subseteq \Gamma(l_3.F.R)$.
Hence $\Gamma(l_7) \subseteq \Gamma(l_3.F.L)$ and $\Gamma(l_8) \subseteq \Gamma(l_3.F.R)$,
so we can apply the Substitution Lemma to prove Condition 1.
\end{proof}

\begin{theorem}[SF Evaluation Coherence]
For any reduction in context $C[t^{l_1}]^{l_2} \rightarrow C[t'^{l'_1}]^{l'_2}$,
if $\Gamma, \varphi \models C[t^{l_1}]^{l_2}$ then
\emph{1)} $\Gamma, \varphi \models C[t'^{l'_1}]^{l'_2}$ and
\emph{2)} $\Gamma(l'_2) \subseteq \Gamma(l_2)$.
\end{theorem}
\begin{proof}
The proof is as for SK-calculus.
The only point of note is that
the constraints between 
the application at the hole of the context and $t'$
now include a constraint on an abstract $@$ value.
However, this is still handled by using the Substitution Lemma.
\end{proof}

\begin{corollary}[SF Soundness]
If $\Gamma, \varphi \models t^l$ and $t \rightarrow^* t'$ then $\Gamma, \varphi \models {t'}^{l'}$.
\end{corollary}
\begin{proof}
As for SK-calculus.
\end{proof}

\section{Evaluation}
It is currently difficult to evaluate meaningfully the usefulness of this analysis.
If one wishes to evaluate an analysis for untyped $\lambda$-calculus,
then by using the usual Church encodings for numbers, lists and other datatypes,
one can easily test it against examples from any textbook on functional programming.
Similarly, using the translation $\unlambdat$,
it is not much harder to evaluate an analysis for SK-calculus in this way.

There is a straightforward translation from SK-calculus to SF-calculus:
simply replace $K$ with $F F$.
It is easy for our analysis to determine that the only possible first argument
to the first $F$ is just $F$, and hence that it will never be factorable.
This activates constraints that are very similar to those for $K$
in the analysis for SK-calculus.
Thus it makes no difference to the precision of 0CFA
whether it is done on a term of SK-calculus
or the same term translated into SF-calculus.

While this is encouraging in that it suggests it is reasonable
to refer to our analysis as 0CFA,
it does not really tell us anything interesting about the precision of the analysis.
The translated program does not use the power of factorisation in a meaningful way,
or indeed (considering that only one reduction of $F$ is used) at all.
There is no interesting suite of programs written in SF-calculus
against which to test the analysis;
nor is there any existing idiomatic translation from any higher level language to SF-calculus.

If we consider only programs that do not deconstruct code
(such as straightforward translations of SK-calculus programs),
our analysis has the same strengths and weaknesses as other forms of 0CFA:
it can analyse some higher order control flow within a program,
but loses precision when the same function is used in two different contexts.

If we consider programs that do
inspect and manipulate the internal structure of code,
there are three further places where we can lose precision.
Firstly, we cannot always tell whether an argument to $F$ will be factorable or not
and in this case, we over-approximate its behaviour to cover both cases.
Our technique essentially works by tracking how many arguments a combinator has been given.
This is unlikely to work well when a term is simultaneously used recursively and partially applied.
Secondly, when we abstractly factorise a term, we lack any contextual information,
so if two applications flow into the same factorisation, we will conflate their factors.
This is similar to the imprecision introduced by lack of context
when using the same function in two different places in ordinary 0CFA.
Finally, while we make a reasonable attempt to track reduction of a term
for the purpose of determining whether its normal form is an atom,
we have no way of discarding non-normal forms when we factorise abstractly,
so we may consider the factorisation of terms that are not factorable forms.

\section{Related Work}

0CFA and other forms of control flow analysis have been widely studied;
see the work of Midtgaard~\cite{DBLP:journals/csur/Midtgaard12} for a detailed survey.

To our knowledge, this is the first static analysis for SF-calculus.
There has been some work on analysing other styles of metaprogramming.
For example, Choi and others~\cite{DBLP:conf/popl/ChoiAYT11} consider how to analyse
a form of extensional metaprogramming called staged metaprogramming,
which captures the composition of code templates.
They suggest using an unstaging translation that turns
the metaprogramming constructs into function abstraction and record lookup,
then using other existing analyses.
Our own work considers how to formulate 0CFA in a dynamically typed language
with staged metaprogramming and variable capture~\cite{DBLP:conf/csfw/0001OS13},
with a view to analysing JavaScript's eval construct~\cite{DBLP:conf/pldi/Lester13}.
In a statically typed setting, Berger and Tratt develop a Hoare-style logic~\cite{DBLP:conf/lpar/BergerT10}
for a language with staged metaprogramming.

Intensional metaprogramming has often been ignored because of its semantic difficulties,
or because of the dominance of the idea
that extensionally equal programs ought to be indistinguishable~\cite{DBLP:journals/corr/JayV14}.
ReFLect~\cite{DBLP:journals/jfp/GrundyMO06},
a functional programming language for hardware design and theorem proving,
allows deconstruction of code values,
but this causes difficulties for its type system,
even in a combinatory fragment of the language~\cite{DBLP:journals/corr/MelhamCC13}.

The idea that program code can be deconstructed
and that its structure can influence the control flow of a program
is conceptually similar to the functional programming idiom
of defining functions by pattern-matching over algebraic datatypes.
There has been some work on analysing functional programs from this perspective.
For example, Jones and Andersen present an analysis
that uses tree grammars to over-approximate the structure of
data values that may be produced by a program~\cite{DBLP:journals/tcs/JonesA07}.
Ong and Ramsay suggest a formalism called Pattern Matching Recursion Schemes
that captures the idea in a typed setting and
develop a powerful analysis for it~\cite{DBLP:conf/popl/OngR11}.

\section{Future Work and Conclusions}

We have presented the first static analysis for SF-calculus,
a formalism which presents a promising foundation for writing programs that transform other programs.
We have proved correctness of the analysis and shown that is comparable to standard 0CFA
for programs that do not rely on the ability of $F$ to factor terms,
such as those translated directly from SK-calculus.
From here, there are a number of obvious directions in which to proceed.

Firstly, in order to evaluate the usefulness of the analysis
and to advance our understanding of program transformation,
it would be good to develop a translation from a higher level language that
supports intensional metaprogramming into SF-calculus.
The translation should map code deconstruction to factorisation using $F$.

Secondly, there is scope to improve the precision of the analysis.
For standard 0CFA, tracking context in the style of $k$-CFA
or a pushdown analysis in the style of CFA2 can improve precision significantly.
The same techniques may be applicable here.
It may also be possible to use techniques
from analysing pattern matching and tree datatypes in functional programming languages
to analyse the term trees that constitute programs
in SF-calculus and their pattern-matching and deconstruction with $F$.
However, an important consideration in applying any such technique to SF-calculus
would be the need to distinguish between a non-factorable term $t$
and the factorable term $t'$ to which it may reduce.

Finally, 0CFA is often useful not as an end to itself,
but because it can be combined with other analysis techniques,
for example drawn from abstract interpretation,
in order to improve their precision by reducing the number of
execution paths or reduction sequences that must be considered
to over-approximate the behaviour of a program.
It would be interesting to see if,
combined with such techniques,
this analysis can actually be used to verify properties of
programs that perform program transformations.

\bibliographystyle{eptcs}
\bibliography{refs}
\end{document}